\documentclass[11pt]{article}

\usepackage{setspace,graphicx,epstopdf,amsmath,amsfonts,amssymb,amsthm,versionPO}
\usepackage{marginnote,datetime,enumitem,subfigure,rotating,fancyvrb}
\usepackage{hyperref,float}
\usepackage[longnamesfirst]{natbib}
\usdate

\excludeversion{notes}		
\includeversion{links}          

\iflinks{}{\hypersetup{draft=true}}

\ifnotes{%
\usepackage[margin=1in,paperwidth=10in,right=2.5in]{geometry}%
\usepackage[textwidth=1.4in,shadow,colorinlistoftodos]{todonotes}%
}{%
\usepackage[margin=1in]{geometry}%
\usepackage[disable]{todonotes}%
}

\makeatletter\let\chapter\@undefined\makeatother 

\newcommand{\Qbb}{\ensuremath{\mathbb{Q}}{}\!}

\newcommand{\cadlag}{c\'adl\'ag}

\newcommand{\ladc}{{\it ladc}}
\newcommand{\ledc}{{\it l\'edc}}

\newcommand{\E}{{\mathbb E}}
\newcommand{\R}{\ensuremath{\mathbb R}}

\newcommand{\Pbar}{\ensuremath{\overline{P}}}

\newcommand{\I}[1]{\ensuremath{\mathbb I}_{#1}}

\newcommand{\cds}{\ensuremath{\text{CDS}}}
\newcommand{\premium}{\ensuremath{\text{PremiumLeg}}}
\newcommand{\prot}{\ensuremath{\text{ProtectionLeg}}}
\newcommand{\lgd}{\ensuremath{\text{L{\footnotesize\textsc{GD}}}}}

\newcommand{\CallZCB}{\ensuremath{\text{Call{ZCB}}}}
\newcommand{\CallCDS}{\ensuremath{\text{Call{CDS}}}}
\newcommand{\one}{\ensuremath{\mathbb I}}

\newcommand\beq{\begin{equation*}}
\newcommand\eeq{\end{equation*}}
\newcommand\ben{\begin{equation}}
\newcommand\een{\end{equation}}
\newcommand\bea{\begin{align*}}
\newcommand\eea{\end{align*}}
\newcommand\bean{\begin{align}}
\newcommand\eean{\end{align}}

\newcommand{\aproof}[1]{\begin{proof}{#1:}}
\newcommand{\Halmos}{}
\newcommand{\aendproof}{\end{proof}}

\usepackage{indentfirst} 
\usepackage{endnotes}    
\usepackage{jf}          
\usepackage[labelfont=bf,labelsep=period]{caption}   
\newtheorem{theorem}{THEOREM}
\newtheorem{definition}{DEFINITION}

\newtheorem{proposition}{PROPOSITION}

\begin{document}

\setlist{noitemsep}  
\onehalfspacing      

\author{CHRIS KENYON and ANDREW GREEN\thanks{\rm Both authors: Lloyds Banking Group, 10 Gresham Street, London EC2V 7AE. 
{\bf Disclaimer:} The views expressed in this work are the personal views of the authors and do not necessarily reflect the views or policies of current or previous employers. Not guaranteed fit for any purpose.  Use at your own risk.  {\bf Acknowledgements:}  The authors gratefully acknowledge feedback from the participants at the 2015 PRMIA conference ``Limitations of Current Hazard Rate Modelling in CVA'', and useful discussions with Roland Stamm.}}

\title{\Large \bf Dirac Processes and Default Risk}

\date{17 April 2015}              


\maketitle
\thispagestyle{empty}

\bigskip

\centerline{\bf ABSTRACT}

\begin{doublespace}  
  \noindent We introduce Dirac processes, using Dirac delta functions, for short-rate-type pricing of financial derivatives.  Dirac processes add spikes to the existing building blocks of diffusions and jumps.  Dirac processes are Generalized Processes, which have not been used directly before because the dollar value of non-Real numbers is meaningless.  However, short-rate pricing is based on integrals so Dirac processes are natural.  This integration directly implies that jumps are redundant whilst Dirac processes expand expressivity of short-rate approaches.  Practically, we demonstrate that Dirac processes enable high implied volatility for CDS swaptions that has been otherwise problematic in hazard rate setups.
\end{doublespace}

\medskip

\noindent JEL classification: G12, C63.

\clearpage

\noindent 
Here we introduce {\it Dirac processes}, built from sequences of Dirac delta functions \citep{Dirac1926a,Hoskins2009a,Duistermaat2010a}, and apply them to short-rate-type pricing of financial derivatives.  Dirac processes expand the set of building blocks in mathematical finance beyond diffusions and jump processes\footnote{Note that mean reverting processes are neither L\'evy processes nor Sato processes \citep{Applebaum2009a,Sato2013a} but can be built from them \citep{schoutens_levy_2009,Kokholm2010a}.} built from Weiner and Poisson processes, to additionally include {\it spikes} based on the Dirac delta function.    Dirac processes are a subset of Generalized Processes which were first introduced by It\^o  \citep{Ito1954a,Gelfand1955a} some ten years after his introduction of It\^o's Lemma \citep{Ito1944a,Doeblin1940a}, but not previously used directly in mathematical finance.   Generalized Processes have not previously been used directly in mathematical finance because the dollar value of values that are not in \R, for example the value of a Dirac delta function at the origin, is meaningless.  However, because short-rate pricing is based on integrals, Generalized Processes, and in particular the Dirac processes we introduce here, are natural.  Dirac processes solve the general problem in short-rate pricing of how to introduce jumps in derived forward rates. This provides a required expressivity to short-rate-type modelling that puts it on a more equal footing with Forward-rate-type modelling which has used jumps for a long time \citep{Cont2003a,Eberlein2006a,Jiang2009a,Crosby2008a,Peng2008a}.  For a particular example we show how to obtain high implied-volatility in options on credit default swaps (CDS swaptions) that is otherwise problematic in hazard-rate models \citep{Jamshidian2004a,Brigo2006a,Kokholm2010a,Brigo2010el,Roti2013a,Weckend2014a,Stamm2015a}.  

Pre-crisis interest rate modelling was driven by the need to price hybrids and exotics leading to an emphasis on tenor Forward rate models such as SABR \citep{Hagan2002a} and Levy-based models \citep{Eberlein2006a}.  In the credit space the emphasis was on structured credit \citep{Brigo2010c} and forward-based modelling hardly developed \citep{Peng2008a}.  Single-name credit modelling stayed with short-rate approaches for the most part \citep{Brigo2006a}.  Post-crisis, funding and credit costs, and regulatory capital costs, have driven the rapid development of valuation adjustments collectively called XVA \citep{Kenyon2012a,Green2014b,Green2015a,Stamm2015a}, and XVA desks, as part of the need to price and manage counterparty netting sets together, mostly based on Monte Carlo techniques.  This has brought short-rate and HJM modelling back into focus for interest rates.  The Dirac processes we develop here fundamentally expand the expressivity of the short-rate and hazard rate approaches by delivering jumps in forward rates and enabling controllable smiles.  They are equally applicable to commodities.

\begin{figure}[t]
\centering
\includegraphics[trim=0 0 0 0,clip,width=0.6\textwidth]{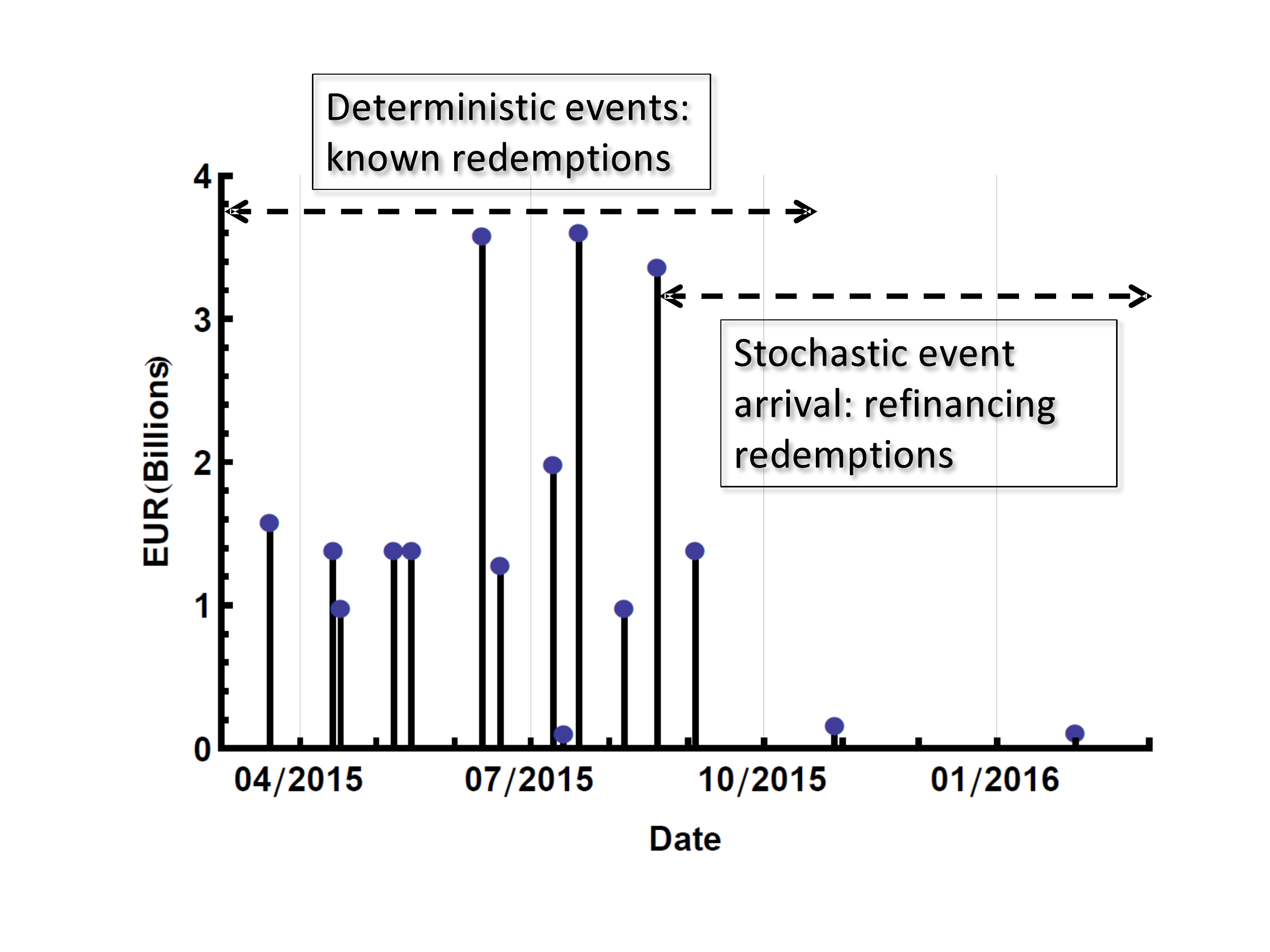}
\caption{Example reference entity: known bond interest and principal repayments and possible refinancing repayments.  Default may be unlikely outside of the specific redemption dates (or the previous weekends).}
\label{f:repay}
\end{figure}
Using a Dirac process for the hazard rate in CDS pricing means that there can be jumps in the survival probability of the reference entity from one day to the next.  Given that the reference entity will generally only default from failing to make payments this provides an explicit model.  For example, as in Figure \ref{f:repay}, it may have to redeem bonds on given dates, i.e. pay back the notionals, but have no comparable payments in the date range.  Practically speaking the reference entity may default on the prior Saturday to each payment, on the dates themselves, but have no other significant possibilities of default during the date range.  Thus modelling hazard rates incorporating Dirac processes captures practically relevant dynamics.  In addition, whilst CDS traders may focus on standardised CDS contracts that all, by definition, mature on only four dates per year (the IMM dates, \cite{markit2009a}), credit pricing in general, e.g. for CVA or XVA \citep{Kenyon2012a,Stamm2015a,Green2015a}, requires arbitrary dates.

Econometrically, jumps are observed in interest rates \citep{Das2002a,Johannes2004a,Piazzesi2005a} and Forward-rate models with jumps are common \citep{Jiang2009a}.  There is a wide literature on pricing interest rate options with jumps \citep{Cont2003a} and general (L\'evy) approaches have been productive \citep{Eberlein2006a} in the Forward-rate space.  However, without Dirac processes, jumps in forward rates or swaption prices cannot be created in the short-rate setting.

Most stochastic processes used in derivative pricing are built from combinations of Brownian motions (continuous) and jump processes (discontinuous, i.e. Poisson processes).  These can be characterized as {\it good integrators} which the  Bichteler-Dellacherie theorem \citep{Protter2010a,Bichteler2011a} demonstrates are equivalent to semimartingales.  Dirac processes, however, escape the Bichteler-Dellacherie theorem, and the definition of semimartingale, because whilst Dirac processes are adapted, they are not \cadlag, nor do they take values exclusively in \R.  Instead, Dirac processes are \ledc, (limites \'egale deux c\^ot\'es), whilst mixed Dirac-Jump processes are, when at least one jump coincides with one Dirac delta function, \ladc\ (limites aux deux c\^ot\'es).  A related, but distinct, strand of research deals with functional extensions of It\^o calculus \citep{Cont2010a}.  

Generalised functions have limitations, for example there is no single accepted method to deal with multiplication of two generalised functions although methods exist \citep{Colombeau1992a,Grosser1999a}.  For our purposes we will always be using integrals of Dirac delta functions so this is not a significant limitation.  We defer further theoretical positioning of Dirac processes, and their practical implications, until  we have defined them.

The first main contribution of this paper is the introduction of Dirac processes for short-rate pricing, expanding mathematical finance beyond diffusions and jumps to include spikes based on Dirac delta functions.   The second main contribution is showing how Dirac processes can be used in derivative pricing including pricing CDS swaptions demonstrating high volatilities.  A third contribution is the potential for detection of incomplete markets from single implied volatilities rather than an implied volatility smile.
  
The paper is organized as follows: the first part introduces Dirac processes starting from their motivation in short-rate pricing, we then recall properties of the Dirac delta function; then we define the Dirac process and describe some of the (many) ways this can be used as a building block.  The second part of the paper deals with interest rate derivative and single-name credit derivative pricing building up to pricing CDS swaptions with high implied-volatilities.   Finally we discuss our results and conclude.

\section{Dirac Processes}

We first motivate our invention of Dirac processes starting from short-rate-type pricing, then recall features of the Dirac delta function, and then define the Dirac process and its relation to the broader category of Generalized Processes and stochastic integration.  Given that the Dirac process is a building block in a similar way to Brownian motion or the Poisson process, there are an arbitrary number of ways to build with it.  Equally, there is no need to use Dirac processes exclusively but they can be combined with diffusions and jumps.  We introduce a basic set of constructions that is in no way exhaustive, but some of which are used in the following section on pricing derivatives with Dirac processes.

\subsection{Motivation}

Consider the price of a riskless zero coupon bond $P(t,T)$, at $t$ with maturity $T$, when the interest rate is modelled using the short-rate approach:
\ben
P(t,T) = \E_t^\Qbb\left[ e^{-\int_t^T r(s) ds} \right]   \label{e:ZCB}
\een
$r(t)$ is the short rate, i.e. the instantaneous spot interest rate.  \Qbb\ is the risk neutral measure which we assume exists, and $\E_t^\Qbb$ is the expectation with respect to the information available up to and including $t$ (i.e. the appropriate filtration at $t$) and using \Qbb.  From Equation \ref{e:ZCB} we see that only the integral of the short rate is relevant, this is also the case for a defaultable bond $\Pbar$ whose price is:
\ben
\Pbar(t,T) =\I{\tau>t} \E_t^\Qbb\left[ e^{-\int_t^T r(s) + \lambda(s) ds} \right]
\een
 where the hazard rate is given by $\lambda(t)$, the default time of the reference entity is $\tau$, and $\I{\tau>t}$ is the indicator function for the survival of the reference entity.  Again, only the integral of the hazard rate is relevant.  In general observable and tradables only use integrals of short rates and hazard rates\footnote{We sometimes use short rate for the approach, which encompasses any asset class, and sometimes for interest rates --- this should be clear from the context.}.

\begin{figure}[t]
\centering
\begin{minipage}[b]{0.425\linewidth}
\includegraphics[trim=0 0 0 0,clip,width=0.99\textwidth]{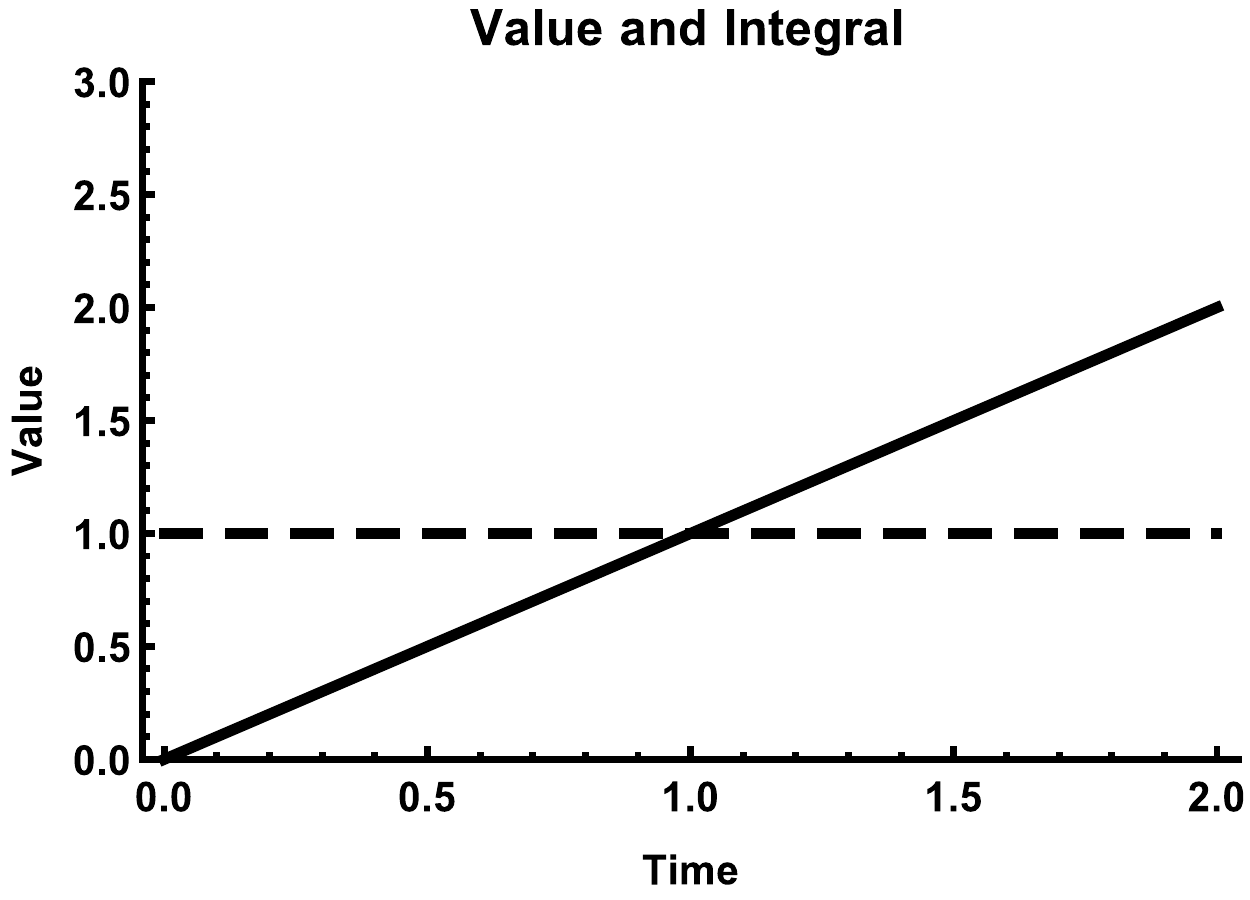}
\end{minipage}
\qquad\qquad
\begin{minipage}[b]{0.425\linewidth}
\includegraphics[trim=0 0 0 0,clip,width=0.99\textwidth]{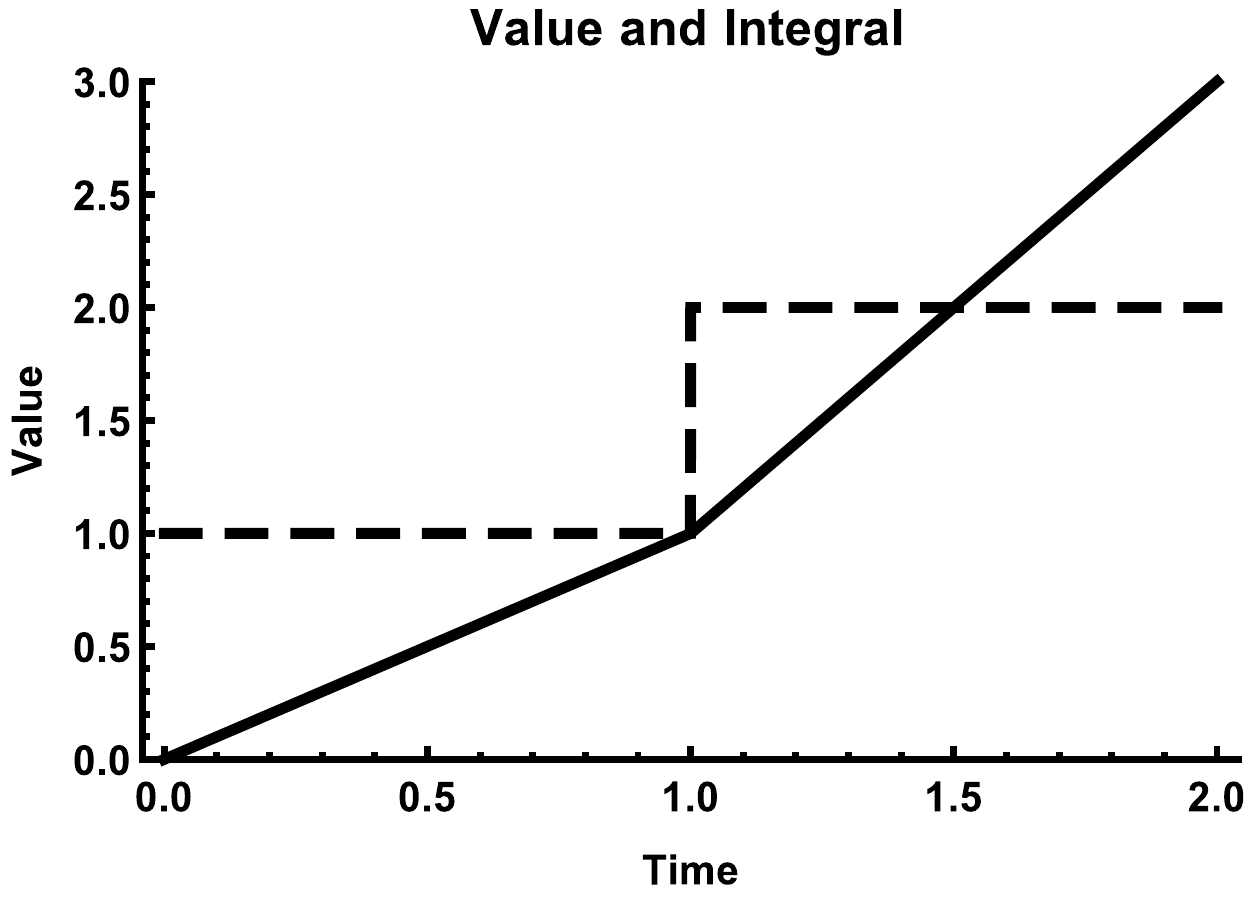}
\end{minipage}
\caption{Processes (dashed) and their integrals (continuous line): (LHS) continuous process ; and (RHS) jump process.}
\label{f:engineering}
\end{figure}
Suppose we want to alter the volatility of these bond prices, we must alter the volatility of the short rate and the hazard rate.  Consider Figure \ref{f:engineering} and the problem from an engineering point of view: what are the effects of introducing a jump process in the short rate on the integral of the short rate (or hazard rate) that determines the prices.  We observe that there is zero significant qualitative effect from introducing a jump process: the integral of the short rate remains continuous.  A non-differentiable point (without Generalized functions) is introduced, but if the short rate was based on a Brownian motion, which is nowhere differentiable (without Generalized Processes), then this is not a qualitative difference.  Clearly, and (in retrospect) obviously, adding jumps to a short rate process is not useful from an engineering point of view to get qualitatively different volatility behaviour.

\begin{figure}[t]
\centering
\includegraphics[trim=0 0 0 0,clip,width=0.425\textwidth]{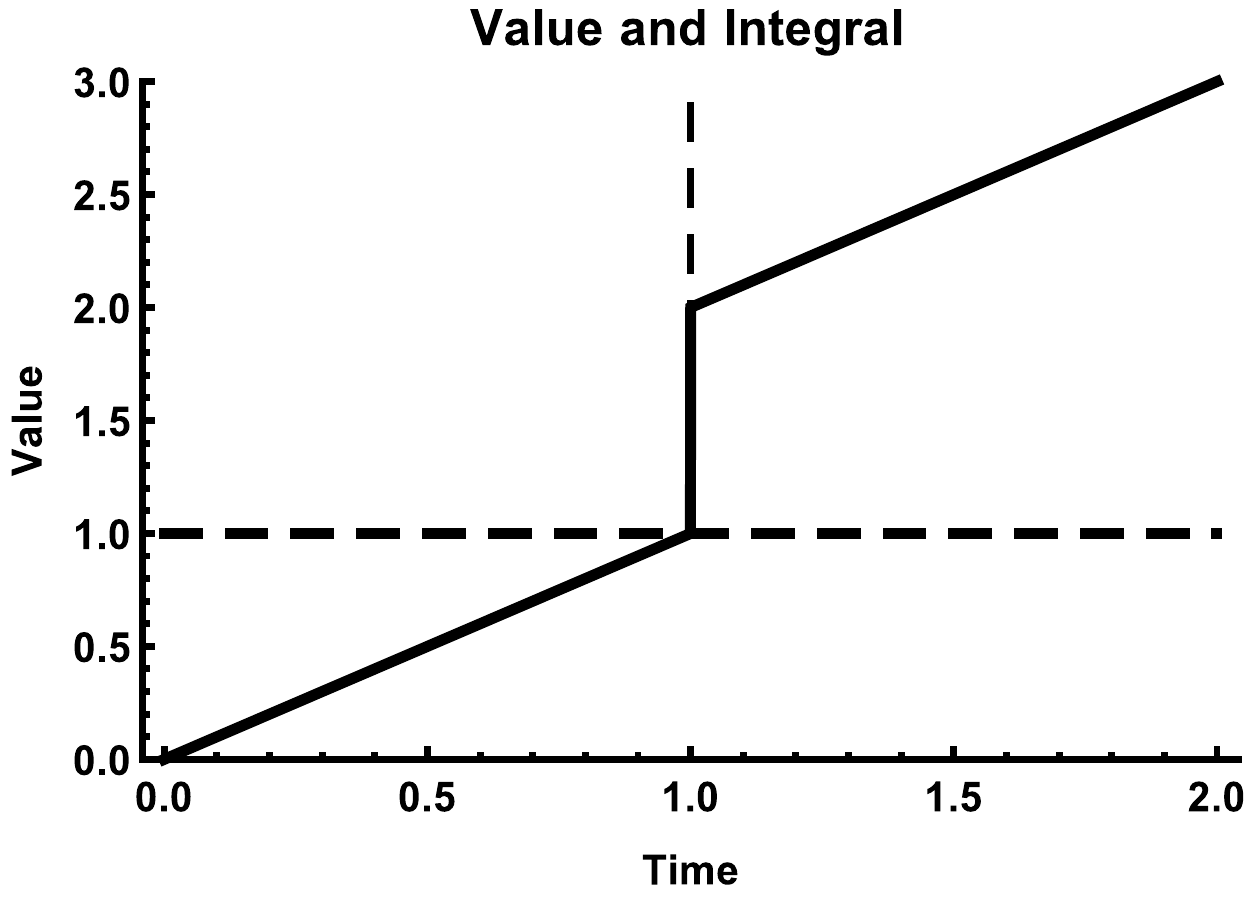}
\caption{Continuous process with unit shifted Dirac delta function (dashed line, including vertical section at $t=1$ extending to infinity but truncated in figure) and integral (continuous line)}
\label{f:engineering2}
\end{figure}
If we add a unit shifted Dirac delta function to the continuous process in Figure \ref{f:engineering} we observe a jump in the integral as shown in Figure \ref{f:engineering2}.  Thus we have achieved qualitatively different behaviour of the integral of the short rate.  Thus a process based on Dirac delta functions is the appropriate engineering solution to constructing qualitatively different volatility in the integral.  

Ignoring the multi-curve nature of post-crisis interest rates for the moment \citep{Kenyon2010a,Mercurio2010a,Moreni2014a} consider a finite tenor forward interest rate from $T_1$ to $T_2$ defined as:
\ben
\frac{1}{T_2-T_1} \left(e^{-\int_{T_1}^{T_2} r(s) ds} -1 \right)  \label{e:fwd}
\een
We see that Dirac delta functions in the short rate also create jumps in the forward rate which are otherwise impossible using only diffusions and jump processes.  Now if we adopt a discount and tenor short-rate approach as in \cite{Kenyon2010a} (rather than the discount plus spread version in \cite{Kenyon2012a}) the price of the Forward Rate Agreement is
\ben
F(t,T_1,T_2) = E_t^\Qbb\left[ e^{-\int_t^{T_2} r^{\text{discount}}(s) ds} \frac{1}{T_2-T_1} \left(e^{-\int_{T_1}^{T_2} r^{\text{tenor}}(s) ds} -1 \right)\right]     \label{e:fwdPrice}
\een
where ``discount'' indicates use of the discount-short rate and ``tenor'' indicates the use of the tenor discount rate, e.g. that for  three month USD Xibor.  Jumps in the Xibor rate dynamics can be present in the discount, the tenor, or both.

A directly analogous setup for commodity forward rates using discount and convenience yields is possible in, say, oil forwards \citep{Brigo2008a} to introduce jumps there.

Let us now turn to default timing.  Consider Figure \ref{f:repay} which shows an example of a known repayment schedule for bond principal and interest.  In the short term, say up to six months in this example, there are many known repayments, but fewer thereafter.  In general this pattern does not mean that the reference entity becomes debt-free but simply that the debt refinancing has not yet been arranged.  Thus we have a mixture of deterministic times with stochastic event outcomes (default or not) and a region of stochastic event times and outcomes.  Many reference entities share this pattern of known and uncertain events.  The key point is that default is very unlikely outside of the repayment dates (and the previous weekends for non-sovereign borrowers) for the first six months.  Most defaults of significantly-sized entities are liquidity defaults (i.e. failure to refinance) not insolvency defaults (declaration by auditors).  This mix of event types can be modelled using a Dirac process with mixed deterministic and stochastic event arrivals.

\subsection{The Dirac delta function}

We provide the basic details of Dirac delta functions here for convenience based on \citep{Arfken2012a,Hoskins2009a,Duistermaat2010a}.  The Dirac delta function was first introduced in its modern form for quantum mechanics \citep{Dirac1926a}  and is widely used in physics (partial differential equations, Green's function, Fourier analysis, \cite{Arfken2012a}). A Dirac delta function is an example of a Generalized function \citep{Hoskins2009a}, also known as distributions \citep{Duistermaat2010a}.  A rigorous mathematical treatment based on distributions\footnote{This re-uses the word from probability with some overlaps.} was built in the 1950s and extended to algebras (i.e. treating multiplication of generalized functions by each other) in the late 1980s \citep{Colombeau1992a,Grosser1999a}.   

Theoretically, generalized functions can be viewed as the smallest closure of the space of usual functions (i.e. those with values in \R) under differentiation \cite{Duistermaat2010a}.    They are not ordinary functions in that their properties can generally only be understood in the context of their effects on other functions, termed test functions, or upon integration to create usual functions.

A Dirac delta function can be considered as a limit of a sequence of functions:
\beq
f_k(t-a) = 
\begin{cases}
1/k & a\le t \le a+k\\
0 & \text{otherwise}
\end{cases}  \label{e:seq}
\eeq
The particular sequence of functions is not unique, many other sequences are possible (e.g. Figure \ref{f:single}), but the properties of the limit are unique in an appropriate sense \citep{Hoskins2009a,Duistermaat2010a}, for example with respect to their behaviour on a universe of test functions.

\begin{figure}[t]
\begin{center}
	\includegraphics[trim=0 0 0 0,clip,width=1.00\textwidth]{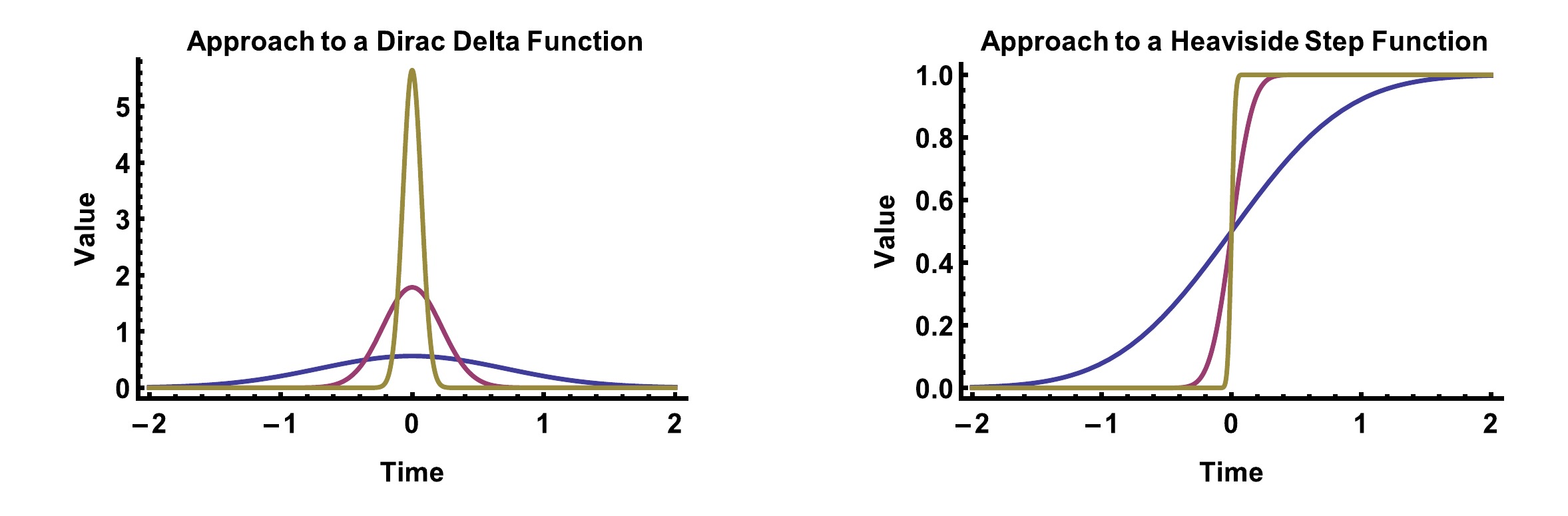}
	\caption{Sequence of functions approaching a Dirac delta function (LHS) and its derivative the Heaviside function (RHS).}\label{f:single}
\end{center}
\end{figure}

Now the integral is constant, i.e.:
\ben
I_k = \int_0^\infty f_k(t-a) dt = \int_a^{a+k} \frac{1}{k} dt = 1 \label{e:dirac0}
\een
Thus the Dirac delta function can be defined as:
\ben
\delta(t-a) = \lim_{k\rightarrow 0} f_k(t-a) \label{e:dirac1}
\een
Obviously this is not a usual function.  Considered as a usual function Equation \ref{e:dirac1} implies that the function is zero everywhere except at $a$ where it is infinity, which is not an element of the Real numbers.

Apart from the integral property in Equation \ref{e:dirac0}, the other key property we will make use of is its sifting property (with no ``h'') on any ordinary function $g(t)$:
\ben
\int_0^\infty g(t) \delta(t-a) dt  = g(a)  \label{e:sift}
\een
The usual distributive property of differentiation holds, i.e.
\beq
(g(t) \delta(t-a))' = g(t)' \delta(t-a) + g(t) \delta(t-a)'
\eeq
where
\beq
\int_0^\infty g(t) \delta(t-a)'  dt  =\ -g(a)'
\eeq
We will not consider multiplication of generalized functions as this is not needed for our applications.

The Dirac delta function can also be considered as the differential of the Heaviside step function, i.e. formally
\ben
\frac{d}{dt}u(t-a) \equiv \delta(t-a) \equiv \delta_a(t)		\label{e:equiv}
\een
where $u(t)$ is the Heaviside step function, $u(t) = 0,\forall\ t<0;\ 1,\forall\ t\ge 0$.

Let us assue that $g(t)$ is some well-behaved test function, where well-behaved means that it takes values in \R\ and is bounded. Using the Riemann-Stieltjes integral we see that integrating this test function multiplied by a  Dirac delta function is equivalent to integrating the test function with respect to a Heaviside step function, i.e.
\beq
\int_{-\infty}^\infty g(t) \delta(t-a) dt = \int_{-\infty}^\infty g(t) du(t-a) = g(a)
\eeq
This provides another connection between Dirac processes and jump processes, and an alternate definition of the Dirac delta function.  That is, since integrals of Dirac delta functions are well-behaved they have a representation in terms of ordinary functions, i.e. as step functions.  We shall see that the derivative of a Poisson process is a Dirac process and vice versa.

A Dirac delta function can also be regarded as a functional, rather than a Generalized function, for its effects on test functions within integrals \cite{Hoskins2009a,Lighthill2003a}.  

\subsection{Dirac processes}

We start by defining a Dirac process and then build upon it.  Just as a Brownian motion alone is not particularly useful, similarly a Dirac process is useful for what can be done with it.

\begin{definition}[Dirac process] A Dirac process $D(t)$ is a sequence of Dirac delta functions $\delta(*)$ separated (time shifted) by exponentially distributed waiting times $e_j$, i.e.
\begin{center}
\fbox{
\begin{minipage}[b]{0.7\linewidth}
\begin{align*}
D(t) :=& \sum_{i=0}^{\infty} \delta(t - E_i)\\
E_i =& \sum_{j=0}^{j=i} e_j  \\
e_j \thicksim & \text{i.i.d. Exponential random variable, mean}\ \nu
\end{align*}
\end{minipage}
}
\end{center}
\end{definition}

\begin{proposition}[Dirac process properties] A Dirac process is Markov, memoryless, and non-anticipative.
\end{proposition}
\aproof{Proof of Dirac process properties}
Since the Dirac delta function exists only at a single time point, and the time between Dirac delta functions is exponentially distributed and the exponential distribution is memoryless all three properties follow.
\Halmos
\aendproof
The proposition above relies on the fact that the Dirac delta function is the limit of the function sequence (as in Equation \ref{e:seq} or similar) rather than any element of that limit.  If this were not the case then there would be a potential issue with being able to predict some time in advance when the next spike would occur.  However, this is not the case.

The rate of arrival of spikes in a Dirac process is constant with rate $\nu$.

\begin{definition}[Non-homogeneous (or inhomogeneous) Dirac process] A Non-hom\-og\-en\-eous Dirac process has a rate parameter $\nu(t)$ that is a deterministic function of time.
\end{definition}

\begin{definition}[Compound Dirac process] A Compound Dirac process $D_x(t)$ is a Dirac process where each spike is scaled by a non-anticipative stochastic process $x(t,i)$
\beq
D_x(t) := \sum_{i=0}^{\infty} x(i,t)\ \delta(t - E_i)
\eeq
\end{definition}
Note that the scaling process,  $x(i,t)$, above does not need to be Markov and does know which spike it is scaling (via $i$).  The effect of the scaling process is to change the integral of the process as in Equation \ref{e:sift}.  We have also not ruled out that the scaling is affected by the time at which the spike occurs ($t$).

\begin{definition}[Deterministic Dirac process] A Deterministic Dirac process is a sequence of time-shifted Dirac delta functions where the time-shifts are known in advance.
\end{definition}

\begin{definition}[Compound Deterministic Dirac process] A Compound Deterministic Dirac process is Compound Dirac process where the time-shifts are known in advance.
\end{definition}

A Compound Deterministic Dirac process models a known sequence of events where there is uncertainty as to the magnitude of the events, but not when they should occur.   In interest rate modelling, for example, the meeting dates of the Bank of England's Monetary Policy Committee (MPC) are known\footnote{\url{http://www.bankofengland.co.uk/publications/Pages/news/2014/119.aspx}} as are the meeting dates for the Federal Open Market Committee (FOMC)\footnote{\url{http://www.federalreserve.gov/monetarypolicy/fomccalendars.htm}} at least a year in advance.  Jumps have been associated with these meeting dates \citep{Piazzesi2005a}.  In credit as in Figure \ref{f:repay},  repayment dates may be known but not the precise probability of default on each repayment date.  In the longer term, for credit, there are repayments from new refinancing so {\it both} a Compound Deterministic Dirac process and a Compound Dirac process can be required.  Equally the exact dates of MPC and FOMC committee meetings are not known more than a year in advance.

Now that we have defined the Dirac process and some variants we will place it into context with respect to stochastic integration  from the \cite{Protter2010a} and \cite{Bichteler2011a} perspective in general, and with respect to Poisson processes in particular with the two following theorems.  We will then demonstrate with the third theorem just how flexible the Generalized process framework is.  The first two theorems are new, whereas the third theorem is known but all three are useful for positioning Dirac processes.

\begin{theorem}[Not good Bichteler-Dellacherie integrator] A Dirac process is not a good integrator in the Bichteler-Dellacherie \citep{Protter2010a,Bichteler2011a} sense.
\end{theorem}
\aproof{Proof of Not good Bichteler-Dellacherie integrator} Obvious because it is not a semi-martingale because it is not a usual function, i.e. taking values only in \R.
\Halmos
\aendproof

\begin{theorem}[Integrator-integration equivalence] \label{t:equiv} Integrating with respect to a Poisson process is identical to integrating a Dirac process with respect to time and vice versa.
\end{theorem}
\aproof{Proof of Integrator-integration equivalence} Direct consequence of Equation \ref{e:equiv}.
\Halmos
\aendproof

The equivalence theorem above shows that we have the choice, formally, of integrating {\it with respect to} a jump process or integrating a Dirac process w.r.t. time in terms of modelling.  Generally we chose to integrate Dirac processes as the most direct and natural approach.  This approach also avoids the potential confusion of specifying whether a jump process is the integrand (the function being integrated) or the integrator (the function that the integration is with respect to).  This potential confusion is particularly acute in credit if we are comparing with survival process approaches where there are jumps in the survival process (which is the integral of the hazard rate process). 

\begin{theorem}[Integration w.r.t. a Dirac process]  Integration with respect to a Dirac process defined for each Dirac delta function as
\ben
\int_{-\infty}^\infty g(t) d\delta(t-a) = \int_{-\infty}^\infty g(t) \delta'(t-a) dt = \ -g'(a)
\een
is well defined.
\end{theorem}
\aproof{Integration w.r.t. a Dirac process}  Elementary consequence of Dirac delta function properties and closure of Generalized functions under differentiation.
\Halmos
\aendproof

\begin{figure}[t]
\begin{center}
	\includegraphics[trim=0 0 0 0,clip,width=1.00\textwidth]{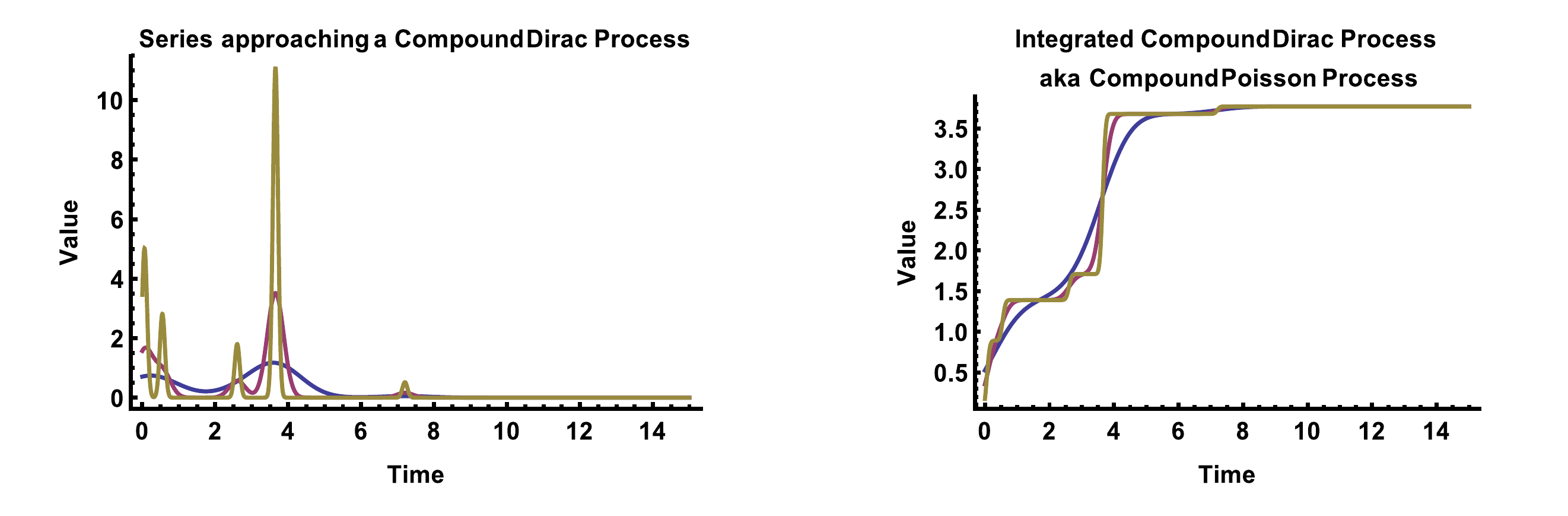}
\end{center}
\caption{Approach to a Compound Dirac process and its integral the Compound Poisson process.}\label{f:compound}
\end{figure}

The three theorems above serve to define integration of Dirac processes and integration by Dirac processes with respect to standard stochastic integration theory \citep{Protter2010a,Bichteler2011a}.  Although a Dirac process is not a good integrator it is possible to integrate with respect to it.  We will only do so formally for convenience, but the machinery is well-defined via Generalized functions.  Integration of, rather than {\it by}, Dirac processes  is also illustrated by the relationship between a Compound Dirac process and its integral the Compound Poisson process as in Figure \ref{f:compound}.

Qualitatively, integration of (not {\it by}) the respective processes moves rightwards below.  
\beq
\text{Dirac\ process} \overset{\int}{\rightarrow} \text{Jump process}  \overset{\int}{\rightarrow} \text{Continuous process} \label{e:procs}
\eeq
It is possible to differentiate Brownian motion in the space of Generalized processes, but this is has not proved generally useful in mathematical finance although Generalized Processes are mentioned in the context of white noise \cite{Seydel2012a}.  It\^o's first approach \citep{Ito1944a} via It\^o integrals is the dominant method in derivative pricing.  It\^o later invented Generalized processes \citep{Ito1954a,Gelfand1955a} but these have not previously found general application in mathematical finance.

\section{Pricing with Dirac Processes}

In the previous section we defined the Dirac process and a set of variants, now we turn to pricing financial derivatives using these Dirac processes.  We will start by considering deterministic Dirac (spike) arrival rates which are sufficient for linear products and then move to stochastic Dirac processes for non-linear products.  We will price an option on a zero coupon bond and then a CDS swaption in a hazard rate setup.  We pick CDS swaptions because they have previously been problematic to price in hazard rate setups even with jumps, i.e. it has been difficult to obtain sufficiently high implied volatilities \citep{Jamshidian2004a,Brigo2006a,Kokholm2010a,Brigo2010el,Roti2013a,Weckend2014a,Stamm2015a}.  We stress that the models presented are the first examples of the use of Dirac processes.  Just as there are an arbitrary number of models that can be built from diffusions and jump processes, the same is true of Dirac processes alone and in combination with diffusions and jump processes.

\subsection{CDS and CDS swaptions}

The USD CDS market coverage grew rapidly from around 200 reference entities in 2002 to around 1600 in 2008 and remained at roughly that size thereafter, see Figure \ref{f:cdscount}.  Coverage is relatively small compared to the number of counterparties of any major bank.
\begin{figure}
	\centering
		\includegraphics[width=0.6\textwidth]{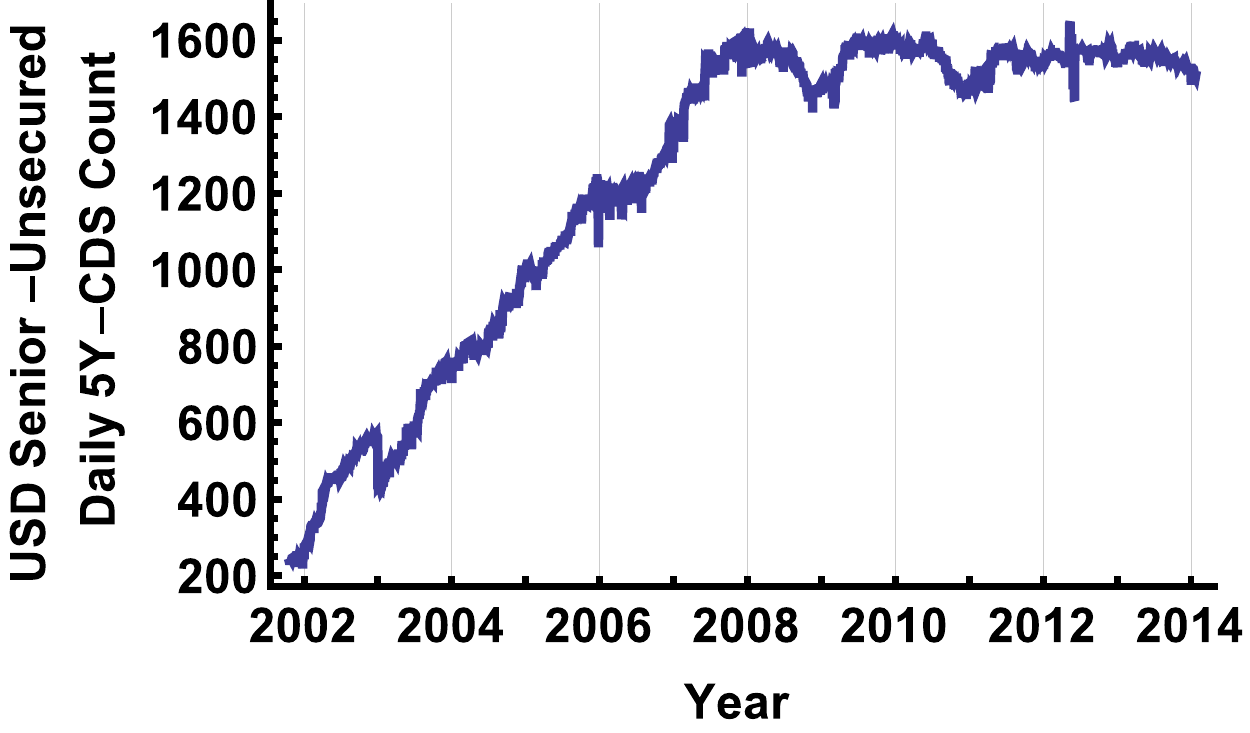}
\caption{USD CDS market size from quotes from a major data provider pre-selected by the provider for at least minimal liquidity.}
	\label{f:cdscount}
\end{figure}
We interpret CDS spreads below in terms of protection and premium legs, as is standard \citep{Brigo2006a}.  However, this ignores the fact that credit protection provides regulatory capital relief \citep{BCBS-189}, which may be significant in interpreting CDS spreads \citep{Kenyon2013d}.  We leave this extension dealing with incorporating capital pricing along the lines of \citep{Green2014b} for future work.

\begin{figure}
	\centering
		\includegraphics[width=0.8\textwidth]{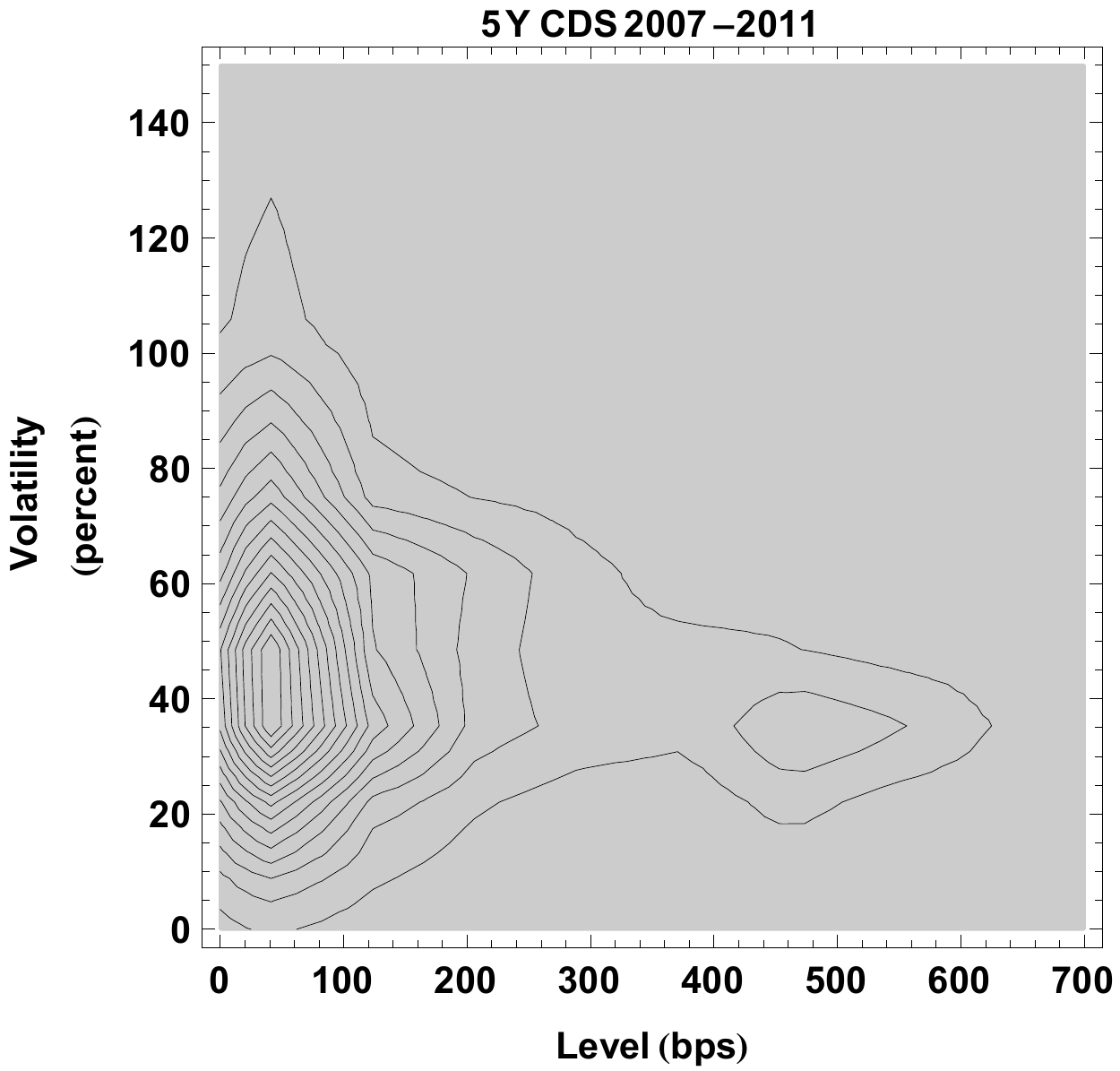}
	\caption{Density plot of historical 5Y CDS relative volatilities and spreads mid-2007 to end-2011 using a one-year window for both.  Plot combines data from 122 reference entities with at least 99\%\ data coverage of the period.  CDS are senior (SNRFOR), USD-denominated, and with XR document clause (most common type).  Volatility is in percent and Level is in bps for the mean CDS spread (each taken from overlapping one-year windows).  This is a downward-biased sample, at least for CDS level, because of the requirement for (almost) complete data across the crisis.}
	\label{f:cdsvol}
\end{figure}

The single-name CDS swaption market is almost exclusively OTC (over-the-counter, i.e. bespoke).  Historical CDS spread volatility can be very high and level dependent, see Figure \ref{f:cdsvol}.  The data in the figure are for reference entities with (almost) complete data series for 5Y CDS quotes from mid-2007 to end-2011 so there is a significant downward bias in the data as quotes dried up on many financial (and other) reference entities during the crisis.  We do observe that it is certainly not true that lower CDS spread always implies lower CDS spread volatility, at least historically.  Individual reference entities display such a wide variety of patterns that the summary across all reference entities is given here.  Market-implied volatilities may be much higher than historical, especially considering CDS index volatilities, e.g. CDX \cite{Stamm2015a} which have crossed 100\%\ several times.  Note that here we focus on single-name credit and leave multi-name credit for future research.

\subsection{Short Rate and Hazard Rate Pricing}

Here we give the model-independent formulae that we will expand upon in later sections. For credit derivatives the short rate approach is usually called the hazard rate approach.

The price of a riskless (i.e. non-defaultable) zero coupon bond is
\beq
P(t,T) := \E_t\left[ e^{-\int_t^T r(u) du} \right]
\eeq
where $r(u)$ is the short rate.  A defaultable zero coupon bond is
\beq
\Pbar(t,T) := \E_t\left[ e^{-\int_t^T r(u) + \lambda(u) du} \right]
\eeq
where $\lambda(u)$ is the hazard rate.  A semi-defaultable zero coupon bond is
\beq
\Pbar(t,T_1,T_2) := \E_t\left[ e^{-\int_t^{T_1} r(u) du -\int_t^{T_2} \lambda(u) du} \right]
\eeq
The semi-defaultable case is where the default risk and the discounting have different maturities.  This is not a (currently) traded instrument but is used below in CDS pricing and is the general case.

A CDS contract provides protection against default of the reference entity in exchange for a series of periodic payments called coupons.  Thus there is a protection leg paying 1-recovery on reference entity default, and a premium leg where coupons are paid provided that the reference entity has not yet defaulted.  We will neglect accruals for simplicity, they can be added simply.  Since the Big Bang\footnote{CDS Market, not the astronomical one.} CDS spreads are quantized, i.e. they take only a certain set of standard values and any difference is made up using an up-front fee.  Many other features are also standardized but these details are not material here.

Default is not, practically, a continuous process but rather only happens with a minimum temporal resolution of about a day, so we use a discrete expression for the protection leg rather than the integral form that is more commonly seen \citep{Brigo2006a}.  This is more accurate and can be expressed in terms of the primitives we have already, i.e. defaultable bonds and semi-defaultable bonds.

The formula below is model independent.

\begin{theorem}[CDS Price]\label{t:cds}  The price of a CDS for protection $[T_a,T_b]$ with premium $R$ is
\begin{align}
\cds_{a,b}(R) &= \prot_{a,b} - \premium_{a,b}(R)  \\
&=   \lgd \sum_{i=c+1}^d \left(\Pbar(0,T_i) - \Pbar(0,T_i,T_{i-1})\right)   -  R\sum_{i=a+1}^b \alpha_i\Pbar(0,T_i)
\end{align}
where $\alpha_i$ is the year fraction, and $[T_c,\ldots ,T_d]$ gives a daily partition of $[T_a,T_b]$, and $\lgd$ is the loss given default.
\end{theorem}
\aproof{Proof}
The Premium leg is the discounted coupon payments conditional on the reference entity not having defaulted.  This is a sum of defaultable bonds scaled with the premium payments.  The Protection leg is the sum of default risks over each day discounted back to time zero.
\Halmos
\aendproof

\subsection{Pricing with Deterministic Intensity and Severity}

This section gives prices for zero coupon bonds, riskless, defaultable, and semi-defaultable with deterministic intensity and severity.  In the next Section we cover the stochastic case and option pricing.

\subsubsection{Zero Coupon Bonds}

Consider a general deterministic Dirac case where we use a differential shorthand for hazard rate $\lambda(t)$ given by the $s$-scaled inhomogeneous Dirac process $D_{\zeta(u)}$:
\beq
d\lambda(u) = s\ dD_{\zeta(u)}(u)
\eeq
where $s$ is the (deterministic) severity of each spike.  This equation states that the hazard rate is zero except where it is a $s$-severity Dirac delta function, and the process is inhomogeneous with (rate) intensity $\zeta(u)$.  It is clear that $\lambda(u)$ is independent of $r(u)$ assuming that the interest rate short-rate\footnote{We will use short-rate for interest rate short-rate hereafter, and hazard rate for default rate short rate.} has no jumps and no Dirac process component.  

\begin{theorem}[Bonds: Dirac Hazard with Continuous Short Rate] The price at time t of a defaultable zero coupon bond maturing at time T with unit notional, assuming
\begin{itemize}
	\item the short rate is continuous;
	\item the hazard rate is given by $d\lambda(u) = s\ dD_{\zeta(u)}(u)$ where $s,\ \zeta(u)$ are deterministic;
\end{itemize}
is given by:
\beq
\Pbar(t,T) = P(t,T) \exp\left(\left( e^{-s} -1 \right) \int_t^T \zeta(u) du  \right)
\eeq
and the price of a semi-defaultable zero coupon bond is given by:
\beq
\Pbar(t,T_1,T_2) = P(t,T_1) \exp\left(\left( e^{-s}  -1 \right) \int_t^{T_2} \zeta(u) du  \right)
\eeq
\end{theorem}
\aproof{Proof}
We will prove the semi-defaultable case as it is the more general.  Since there are no jumps in the riskless bond price the semi-defaultable bond price is independent of it, by independence of Dirac and Brownian processes.  Thus:
\begin{align}
\Pbar(t,T_1,T_2) &=  P(t,T_1)\E_t\left[ e^{-\int_t^{T_2} \lambda(u) du} \right] \\
&= P(t,T_1) \E_t [e^{-s N(\int_t^{T_2} \zeta(u) du)}] \\
&= P(t,T_2) \exp\left(\left(e^{-s}  -1 \right)  \int_t^{T_2} \zeta(u) du  \right)
\end{align}
where $N(\int_t^{T_2} \zeta(u) du)$ is a random variable with a Poisson distribution with scale $\int_t^{T_2} \zeta(u) du$.  
\Halmos
\aendproof

Implicit in these equations is the fact that the event arrival and severity can be traded off against each other.  Making the approximation (for now and for clarity) taken from \cite{Brigo2006a} that
\begin{align}
\text{hazard\ rate} =& \frac{\text{CDS\ spread}}{\text{Loss\ given\ default}} \\
\lambda =& \frac{\cds}{\lgd}
\end{align}
then equating the two expressions for a defaultable zero coupon bond we have
\ben
P(0,t) e^{-\lambda t} = P(0,t) e^{(e^{-s}-1)\nu t}
\een
which implies that the tradeoffs between arrival rate $\nu$ and severity $s$, are
\ben
\nu = \frac{\lambda}{1 - e^{-s}} \label{e:tradeoff}
\een
Figure \ref{f:tradeoff} shows the tradeoffs for two CDS levels.  These tradeoffs will be significant for multi-name credit to change the default-time correlation but since this is out of scope for this paper we do not go further here.
\begin{figure}[t]
\centering
\includegraphics[trim=0 0 0 0,clip,width=0.5\textwidth]{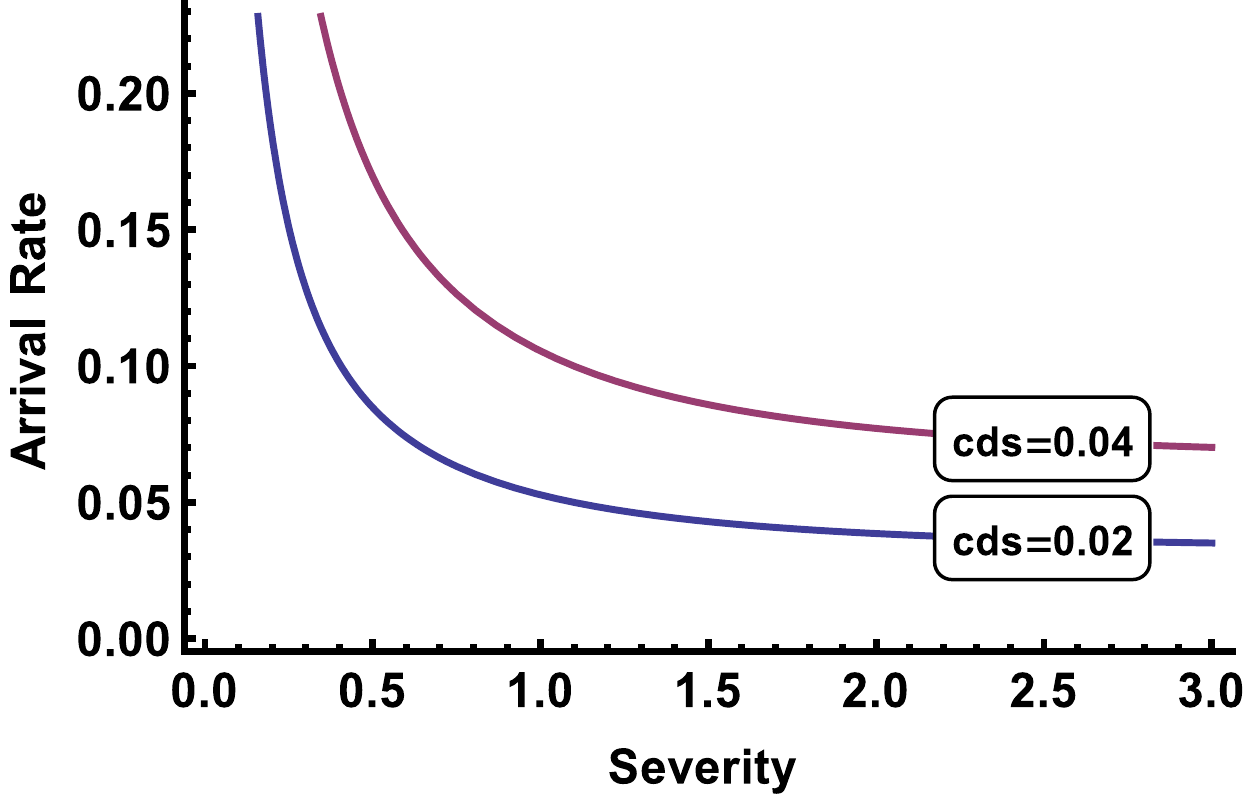}
\caption{Tradeoffs possible between event arrival rate and severity for two CDS spread levels (200bps and 400bps) using Equation \ref{e:tradeoff}, see text for details.}
\label{f:tradeoff}
\end{figure}

We now consider the case where the short rate is also a Dirac process.  After this we will consider the mixed case where the short rate is a mix of continuous and Dirac processes.
\begin{theorem}[Bonds: Dirac Hazard with Dirac Short Rate]\label{t:DHDS} The price at time t of a defaultable zero coupon bond maturing at time T with unit notional, assuming the short rate and hazard rates are:
\begin{align}
d r(u) & = s_1 dD_{\zeta_1(u)}(u) + s_0 dD_{\zeta_0(u)}(u)\\
d\lambda(u) &= s_2 dD_{\zeta_2(u)}(u) + s_0 dD_{\zeta_0(u)}(u)
\end{align}
is given by:
\begin{align}
\Pbar(t,T)) =& \exp\left(
2\left( e^{-s_0}  -1 \right) \int_t^T  \zeta_0(u)du  
+\left( e^{-s_1}  -1 \right) \int_t^T  \zeta_1(u)du   \right.\\
&\left.\qquad {}+\left( e^{-s_2}  -1 \right) \int_t^T  \zeta_2(u)du
 \right)
\end{align}
and the price of a semi-defaultable zero coupon bond is given by:
\begin{align}
\Pbar(t,T)) =& \exp\left(
\left( e^{-s_0}  -1 \right) \int_t^{T_1}  \zeta_0(u)du  
+\left( e^{-s_1}  -1 \right) \int_t^{T_1}  \zeta_1(u)du   \right.\\
&\left.\qquad {}+
\left( e^{-s_0}  -1 \right) \int_t^{T_2}  \zeta_0(u)du  
+\left( e^{-s_2}  -1 \right) \int_t^{T_2}  \zeta_2(u)du
 \right)
\end{align}
\end{theorem}
\aproof{Proof}
Almost surely there are no common spikes between different Dirac processes $\zeta_i(s)$.  Now:
\begin{align}
\Pbar(t,T)) &= \E_t \left[e^{-\left(s_0 N(\int_t^T \zeta_0(u) du) + s_1 N(\int_t^T\zeta_1(u) du) + s_0 N(\int_t^T\zeta_0(u) du) + s_2 N(\int_t^T\zeta_2(u) du)\right)} \right]  
\end{align}
each Poisson process $N(*)$ is independent, so their exponentials are independent because functions of independent random variables remain independent so the result follows.  The semi-defaultable bond is a direct extension. \Halmos
\aendproof

\begin{theorem}[Bonds: Dirac Hazard with Mixed Short Rate]  The price at time t of a defaultable zero coupon bond maturing at time T with unit notional, assuming the short rate and hazard rates are as in Theorem \ref{t:DHDS} except that the riskless bond price is:
\beq
P(t,T) = P_C(t,T)P_D(t,T)
\eeq
where $P_C(t,T)$ is the part of the price given by a continuous short rate and $P_D(t,T)$ the part of the price given by a Dirac process, is given by:
\begin{align}
\Pbar(t,T)) =& P_C(t,T)\exp\left(
2\left( e^{-s_0}  -1 \right) \int_t^T  \zeta_0(u)du  
+\left( e^{-s_1}  -1 \right) \int_t^T  \zeta_1(u)du   \right.\\
&\left.\qquad\qquad\qquad {}+\left( e^{-s_2}  -1 \right) \int_t^T  \zeta_2(u)du
 \right)
\end{align}
and the price of a semi-defaultable zero coupon bond is given by:
\begin{align}
\Pbar(t,T)) =& P_C(t,T_1)\exp\left(
\left( e^{-s_0}  -1 \right) \int_t^{T_1}  \zeta_0(u)du  
+\left( e^{-s_1}  -1 \right) \int_t^{T_1}  \zeta_1(u)du   \right.\\
&\left.\qquad\qquad\qquad {}+
\left( e^{-s_0}  -1 \right) \int_t^{T_2}  \zeta_0(u)du  
+\left( e^{-s_2}  -1 \right) \int_t^{T_2}  \zeta_2(u)du
 \right)
\end{align}
\end{theorem}
\aproof{Proof}
There is no correlation between a continuous process and a Dirac process, so the result follows given Theorem \ref{t:DHDS}.\Halmos
\aendproof

Calibration to yield and survival curves are direct.  We observe, however, that we can trade off the severity against the intensity for the Dirac process but this has no effect on the pricing provided that there is no tradeoff between different Dirac processes.  This could affect multi-name credit pricing but that is out of scope for this paper.

\subsubsection{Interest Rate Swaps}
We confine our attention to interest rate swap (IRS) pricing using a single-curve which we take as riskless.  Typical examples are SONIA, EONIA and Fed Funds curves for GBP, EUR, and USD respectively.  Multi-curve pricing is outside the scope of this paper as the focus is on credit applications.

Given riskless discount bond prices IRS prices are standard based on Equations \ref{e:fwd} and \ref{e:fwdPrice}, making the usual assumption of fully collateralized trades with cash collateral in the same currency.  We neglect initial margin costs for simplicity and any portfolio effects.

\subsection{Pricing with Stochastic Severity: Non-Linear Instruments}

To price options we require stochastic intensities or severities in the Dirac process which we now consider.  There are many possible models and we illustrate a simple and direct approach for stochastic severity, the Dirac-OU-Severity model.

In all models the key is information flow and in this regard a Dirac process poses a particular challenge because it is uniquely memoryless.  When there is no Dirac spike the process is zero.  This is unlike diffusion or jump processes which store some state in their current value.  In a Dirac process this information can be stored in many places but ultimately appears in either the spike intensity, or in the spike severity or both.  This stochastic severity or intensity is driven by an external process, just like the jump size for a jump process.

\paragraph{Example model: Dirac-OU-Severity}
The Dirac-OU-Severity model uses a tanh-trans\-formed Ornstein-Uhlenbeck process for the severity with a deterministic arrival intensity for Dirac spikes:
\begin{align}
\text{Severity driver}\qquad dx &= \theta(\mu-x) dt + \sigma dB  \label{e:ou}\\
\text{Severity}\qquad v &= \frac{\tanh(x)+1}{2} \Big/b + (b-1)/b  \label{e:tanh}\\
&= \frac{e^{2x}}{1+e^{2x}} \Big/b + (b-1)/b
\end{align}  
Equation \ref{e:ou} means that the transition density function is known analytically.  It is conditionally Normal with variance and mean as below
\begin{align}
\text{Var(OU)}(\theta,\sigma,t)=&\frac{\sigma ^2 \left(1-e^{-2 \theta  t}\right)}{2 \theta }\\
\text{Mean(OU)}(x_0,\mu,\theta,\sigma,t)=&\mu  \left(1-e^{-\theta  t}\right)+x_0 e^{-\theta  t}
\end{align}
Equation \ref{e:tanh} is the severity driver transformation, this transforms $[-\infty,\infty]\mapsto [1-1/b,1]$, the severity of the event as shown in Figure \ref{f:band}.  One minus the severity gives the probability of survival given an event.
\begin{figure}[t]
\centering
\includegraphics[trim=0 0 0 0,clip,width=0.5\textwidth]{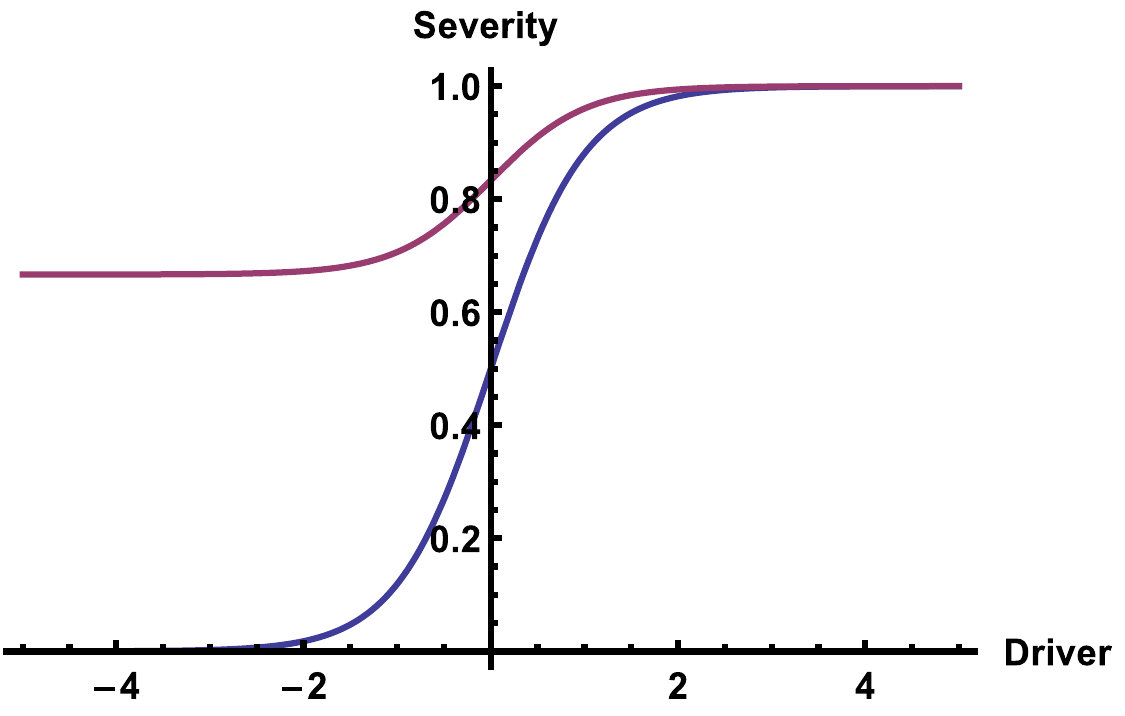}
\caption{Illustration of Equation \ref{e:tanh} showing $b=1$ where the driver maps to $[0,1]$ and $b=3$ where the driver maps to $[3/4,1]$, see text for details.}
\label{f:band}
\end{figure}

We operate in event time for the severity process, i.e. we take a scaling of unit event time between events.  This is a common way of thinking about discrete event processes, like jump process, in financial modelling.  Subordinators generalize the idea of event time, and this model is a simple example in that all times between events have the same calendar distribution (Exponential).

In this Dirac-OU-Severity model the severity is independent of the intensity because the intensity is constant.

With this modelling setup we can price efficiently by discretizing the state space of the severity process.

\subsubsection{Zero Coupon Bonds}

\begin{theorem}[Bonds: Dirac-OU-Severity Hazard with Continuous Short Rate]\label{t:DOUcts} The price at time t of a defaultable zero coupon bond maturing at time T with unit notional, assuming
\begin{itemize}
	\item the short rate is continuous;
	\item the hazard rate is given by
\begin{align}
	d\lambda(u) &= s(u)\ dD_{\zeta(u)}(u)  \\
	s(u) &= (\tanh(x(u))+1)/2\\
	dx(u) &= \theta(\mu-x(u)) du + \sigma dB 
\end{align}
	where $ \zeta(u)$ is deterministic;
\end{itemize}
is given by:
\beq
\Pbar(t,T) = P(t,T)\ h\  \left(\sum_{i=0}^{\infty} (A\ (\one-S))^i P\{N(\int_t^T \zeta(z)dz)) = i\}\right) \one
\eeq
and the price of a semi-defaultable zero coupon bond is given by:
\beq
\Pbar(t,T_1,T_2) = P(t,T_1)\ h\ \left(\sum_{i=0}^{\infty} (A\ (\one -S))^i P\{N(\int_t^{T_2} \zeta(z)dz) = i\}\right)\ \one
\eeq
where: $A$ is the transition matrix; $S$ is the diagonal severity matrix; $h$ is a row vector expressing the state at $t$ (``here''); \one\ is a unit column vector.
\end{theorem}
\aproof{Proof}
We operate in event time, i.e. there is unit event time between events.  Therefore the transition matrix $A$ is constant for each new event.  At each event there is a probability of default, which we can put on the diagonal of the severity matrix $S$.  By assumption the severity is driven by an Ornstein-Uhlenbeck process which is Markov so $S$ is constant.  The combined probability of moving to a new state of severity and survival of the event is $A(\one-S)$ and this is also constant.  

Since the short rate is continuous the riskless bond $P(t,T)$ is independent of the defaultable part.  Now in the interval $(t,T]$ the number of Dirac events is independent of the intensity and so is given by a Poisson distribution with scale $\int_t^T \zeta(z) dz$, hence
\begin{align}
\Pbar(t,T) &=  P(t,T)\ h\ \E_t\left[ (A\ S)^{N(\zeta(T-t))}  \right] \\
&= P(t,T)\ h\  \left(\sum_{i=0}^{\infty} (A\ (\one-S))^i P\{N(\int_t^T \zeta(z)dz)) = i\}\right) \one
\end{align} 
The semi-defaultable bond price follows directly.
\Halmos
\aendproof

Note that in general we have complete freedom to chose $A$ providing it is a valid transition density matrix.  Using an Ornstein-Uhlenbeck process is simply a method of generating $A$ with a small number of parameters.

We now turn to modelling the short rate with a similar process to the hazard rate.  In this case the severity need not be limited to $[0,1]$ but may have an arbitrary range.  For generality we would introduce a scaling factor into the range, but for simplicity of exposition we do not include it.
\begin{theorem}[Bonds: Dirac-OU-Severity Hazard and Mixed Short Rate]\label{t:DOUmixed} The price at time t of a defaultable zero coupon bond maturing at time T with unit notional, assuming
\begin{itemize}
	\item the Dirac part of the short rate and hazard rates are given by:
\begin{align}
d r(u) & = s_1 dD_{\zeta_1(u)}(u) + s_a dD_{\zeta_0(u)}(u)   \label{e:DOU2:DD10} \\
d\lambda(u) &= s_2 dD_{\zeta_2(u)}(u) + s_b dD_{\zeta_0(u)}(u) \label{e:DOU2:DD20} \\
&\zeta_i\ne\zeta_j,\quad i\ne j, \label{e:DOU:ZZ}\\
\text{and}\qquad\qquad&\\
P(t,T) &= P_C(t,T)P_D(t,T)
\end{align}
where $P_C(t,T)$ is the part of the price given by a continuous short rate and $P_D(t,T)$ the part of the price given by a Dirac process
	\item the severity of the three Dirac process are given as below and use Equation \ref{e:tanh} for range transformation,
\begin{align}
dx_a &= \theta_a(\mu_a-x_a) dt + \sigma_a dB_0  \label{e:ouDa}\\
dx_b &= \theta_b(\mu_b-x_b) dt + \sigma_b dB_0  \label{e:ouDb}\\
dx_1 &= \theta_1(\mu_1-x_1) dt + \sigma_1 dB_1  \label{e:ouD2}\\
dx_2 &= \theta_2(\mu_2-x_2) dt + \sigma_2 dB_2  \label{e:ouD3}\\
dB_i dB_j &= \delta_{i,j}dt,\qquad\qquad i,j\in\{1,2,3\} \label{e:DOU2:BB} \\
x_i(t=0) &= 0 \qquad\qquad\qquad i\in\{a,b,1,2\}
\end{align} 
\end{itemize}
then the price of a semi-defaultable zero coupon bond is given by:
\begin{align}
\Pbar(t,T_1,T_2) =& P_C(t,T)\prod_{k\in \alpha}  h_k\ \left(\sum_{i=0}^{\infty} (A_k\ (\one -_k))^i P\{N(\int_t^{T_{\alpha}} \zeta_k(z)dz) = i\}\right)\ \one \label{e:DOU2} \\
& \alpha = \{a,b,1,2\}\\
& T_{\alpha}=\{T_1,T_2,T_1,T_2\}\ \ \text{(Used\ in\ sequence\ with\ $\alpha$).}
\end{align}
where: $A_*$ is a transition matrix; $S_*$ is a diagonal severity (credit) or change (rates) matrix; $h_k$ are a row vectors expressing the state at $t$; \one\ is a unit column vector.
\end{theorem}
\aproof{Proof}
The non-Dirac part of the riskless bonds has no jumps so is independent of all the Dirac terms.  This is now just a generalization of Theorem \ref{t:DOUcts} so we only need to show that the elements in Equation \ref{e:DOU2} are independent.  Note firstly that the driving processes for the severity are independent by construction, Equation \ref{e:DOU2:BB}.  Also the Dirac process in Equation \ref{e:DOU:ZZ} $dD_{\zeta_i},\ i\in{1,2,3}$ are independent by construction.
\Halmos
\aendproof

Since we have the price of defaultable zero coupon bonds we have the price of defaultable coupon bonds.  Since we have the price of semi-defaultable zero coupon bonds we have the price of CDS from Theorem \ref{t:cds}.

\subsubsection{Option Pricing}

Pricing options using Dirac processes is simple, in the Dirac-OU-severity model using the setup above, because we have discretized the state space.  It is further simplified because the different driving processes are independent although they may contribute to both credit and rates parts of defaultable bonds.  

Considering Theorem \ref{t:DOUmixed} and Equation \ref{e:DOU2} pricing an option on a bond is obvious, but long, because we already have the vector $u$ giving the current state in the equation.  We also have the future payoff of $\one$ (one) in all states of the world.  Thus, since we have state-dependent semi-defaultable bond prices we automatically have state-dependent CDS prices and hence options on CDS. 

\begin{theorem}[Bond Option: Dirac-OU-Severity Short Rate]\label{t:DOU1opt} Let the short rate be:
\begin{align}
d r(t)  &= s(t) dD_{\zeta(t)}(t)   \label{e:DOU1only} \\
s(t) &= (\tanh(x(t))+1)/2\\
	dx(t) &= \theta(\mu-x(t)) dt + \sigma dB 
\end{align}
then a European call option with exercise date $T_K$ and strike $K$ on a riskless zero coupon bond with maturity $T$ is given by:
\begin{align}
\CallZCB (t,K,T_K,T) =& h\left( \sum_{i=0}^{\infty} (A\ (\one - S))^i P\{N(\int_t^{T_K-t} \zeta_k(z)dz) = i\} \right)
  \nonumber \\
& \underline{\max}\left(  \sum_{i=0}^{\infty} (A\ (\one - S))^i P\{N(\int_t^{T-T_K} \zeta_k(z)dz) = i\}\one - \underline{K}, \underline{0} \right) \label{e:DOU1opt} \\
& t < T_K\ <\ T\\
& K \ge 0
\end{align}
\end{theorem}
\aproof{Proof}
Both $A$ and $S$ matrices are constant by construction.  The max in the second line of Equation \ref{e:DOU1opt} above gives a column vector of payoffs depending on the state of the world. Then the first line discounts this back to $t$, and the $h$ row vector obtains the price in the current state of the world.
\Halmos
\aendproof

Theorem \ref{t:DOU1opt} shows that option pricing with a single factor has the same complexity as bond pricing.  In contrast when more factors are involved the option payout links the factors on the exercise date and produces a higher complexity.

\begin{theorem}[Bond Option: Dirac-OU-Severity Hazard and Short Rate]\label{t:DOU3opt} With conditions as for Theorem \ref{t:DOUmixed}, except that we take $P_C(t,T)\equiv 0$, the price of a European call option on a semi-defaultable zero coupon bond with strike $K$, exercise date $T_K$ is
\begin{align}
&\CallZCB (t,K,T_K,T_1,T_2) =  h\left(\sum_i \sum_j \sum_k p_1(T_K,i) p_2(T_K,j) p_a(T_K,k) p_b(T_K,k) \right)\nonumber\\
&\times \left(\sum_i \sum_j \sum_k \max(p_1(T_1 - T_K,i) p_2(T_2 - T_K,j) p_a(T_1 - T_K,k) p_b(T_2 - T_K,k) - K, 0)\one \right)\\
&p_{\alpha}(T,i) = (A_\alpha\ (\one - S_\alpha))^i P\{N(\int_t^{T} \zeta_{\beta(\alpha)} (z)dz) = i\}\\
&\alpha\in\{1,2,a,b\}\\
& \beta(\alpha)=\{1,2,0,0\}
\end{align}
\end{theorem}
\begin{proof}
This is a direct generalization of Theorem \ref{t:DOU1opt}
\qed
\end{proof}

Given Theorem \ref{t:cds} and Theorem \ref{t:DOU3opt} it is clear that we can price options on CDS with the different levels of complexity depending on which processes we make deterministic and how many driving processes we use. We can assume as in \citep{Brigo2010el} (based on \cite{Brigo2006a}) that the volatility of riskless rates has little effect and model these as constant.  This is supported here under the assumption that the short rate is continuous.  In this case there is zero correlation with the Dirac process for the hazard rate.  On the other hand if we assume that rates are Dirac and that there are significant common shocks then the assumption of independence is false by construction and there may be significant interaction.  

In any market we may see significant common jumps between large company defaults and interest rates where the company is perceived to have a structural significance.  Equally, if the government shocks the interest rates this may cause a shock to many company integrated hazard, i.e. survival.  Thus it is a matter of calibration, and judgement, how to model the situation in general.  In the interest of space we will only consider the simplest case to demonstrate the capability of Dirac processes.

Few previous models have been able to deliver the very high implied volatilities, around 100\%\, required in the CDS option space \citep{Stamm2015a}.  Here we will show that Dirac-based models can deliver this level of implied volatility.  By implied volatility we mean the volatility backed out from prices using CDS option market models \citep{Brigo2006a}.  These models are direct analogies of Swap Market Models in the interest rate space.

\begin{theorem}[CDS Option: Dirac-OU-Severity Hazard Rate]\label{t:DOUcds1opt} Let the short rate be deterministic and the  hazard rate as before a European call option with exercise date $T_K$ and rate $R$ on a $\text{CDS}_{T_a,T_b}$ with maturity $T$ is given by:
\beq
\CallCDS_{a,b}  = h\ \E\left[\ \Pbar(0,T_K) (\prot_{a,b}(T_K) - \premium_{a,b}(T_K,R))^+ \right] 
\eeq
where the protection and premium legs use the future value of the semi-defaultable bond
\beq
\Pbar(T_1,T_1,T_2) = \left(\sum_{i=0}^{\infty} (A\ (\one -S))^i P\{N(\int_t^{T_2} \zeta(z)dz) = i\}\right)\ \one
\eeq
\end{theorem}
As before underlines indicate vectors, and these are column vectors above.  Underlined operators, i.e. $\underline{\max}$, operate element-wise.
\aproof{Proof}
Now $\Pbar(T_1,T_1,T_2)$ is a vector where the entries correspond to future states of the world.
The result thus follows directly from the definition of the semi-defaultable bond and the discretization of the state space.
\Halmos
\aendproof
Note that pricing an option on a CDS requires marginal effort beyond pricing an option on a zero coupon bond.  This is a direct result of our state-space discretization approach.

\subsubsection{Example CDS Swoption Pricing}

We will now consider an example of CDS swaption pricing to illustrate the high implied volatility possible.  We first note informally that CDS spreads generally move mostly in parallel (as for interest rate yield curves) except when reference entities are distressed when their CDS spread curves invert.  When CDS spread curves invert liquidity moves from the 5Y CDS contract which is usually the most liquid to the 1Y CDS contract and data on longer-dated contracts should be viewed carefully to ascertain whether it is reliable (i.e. executable in comparable volume).  This observation on parallel CDS curve moves suggests that there may be a significant term structure of volatility (TSOV) as in Figure \ref{f:tsov}.
\begin{figure}[t]
\begin{center}
	\includegraphics[trim=0 0 0 0,clip,width=0.50\textwidth]{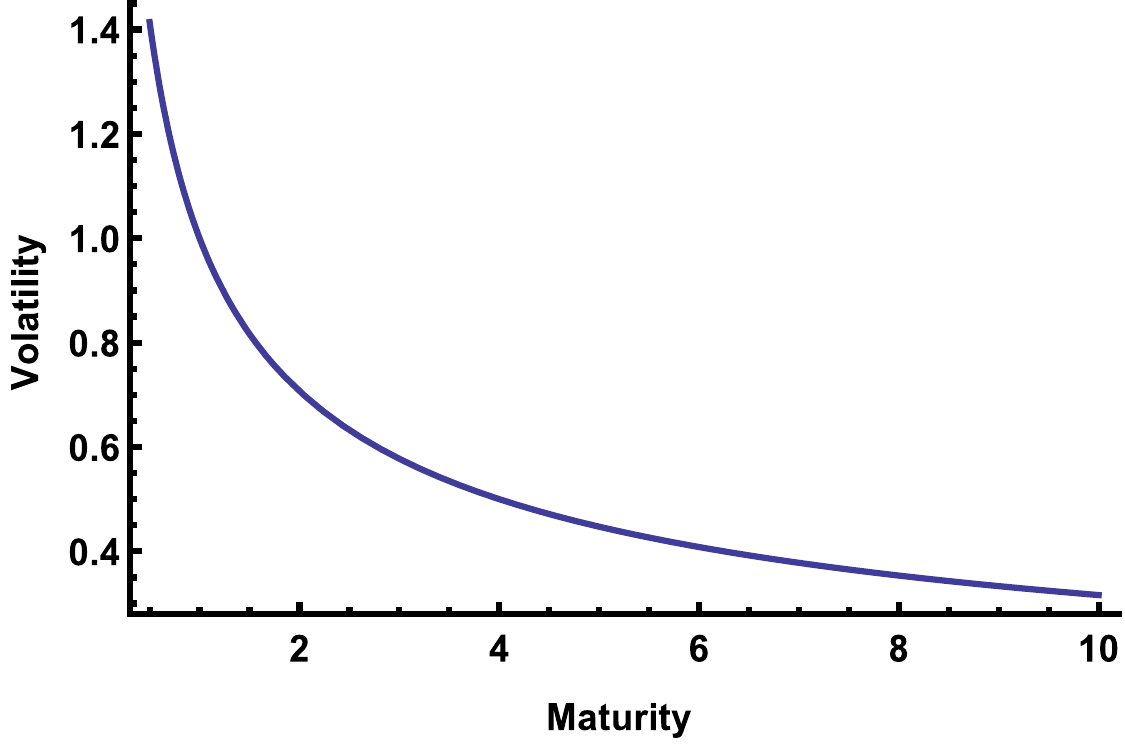}
	\caption{Term structure of volatility shape suggested by observation of mostly parallel moves of CDS curves, i.e. $1/\sqrt{T}$.}
	\label{f:tsov}
\end{center}
\end{figure}
We use a piecewise-flat TSOV with two sections, up to option exercise date and beyond that date ($\sigma_1,\ \sigma_2$), to obtain CDS swaptions with the implied volatilities shown in Figure \ref{f:cdsVol}.  Example CDS swaption parameter settings are below:
\begin{itemize}
	\item $b=6,\ t_{\text{step}}=1,\ \mu=0.73,\ \sigma_1=60\%,\ \sigma_2=10\%,\ \theta=0.1\%,\ \text{recovery}=40\%,\ r=2\%$, intensity$=2$.  
	\item $r$ is the float zero yield curve level.
	\item $t_{\text{step}}$ states that unit event time occurs between events for the severity-driving OU process, and the intensity means that two events are expected to occur per unit calendar time.
\end{itemize}
This model is quite parsimonious with a relatively small number of parameters (six) for the Dirac part. 
\begin{figure}[t]
\begin{center}
	\includegraphics[trim=0 0 0 0,clip,width=0.50\textwidth]{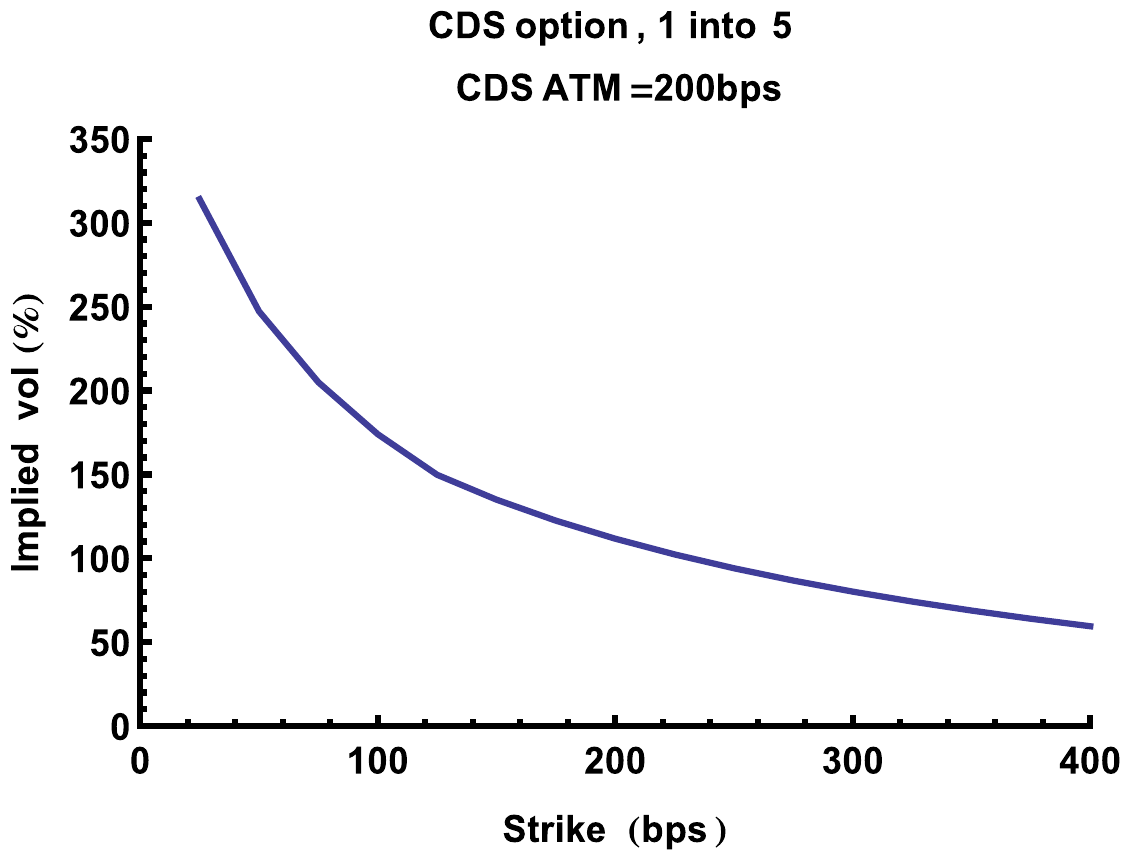}
	\caption{CDS swaption implied volatility smile for 1-into-5 CDS swaption example.}
	\label{f:cdsVol}
\end{center}
\end{figure}

\section{Discussion and Conclusions}

In this paper we have introduced Dirac processes, based on sequences of Dirac delta functions, to mathematical finance, focussing on short rate (interest rates) and hazard rate (credit) models.  Dirac processes provide short-rate-type models with the same expressivity that forward-rate (either instantaneous, HJM, or tenor, LMM) already posses when these include jumps.  In fact every forward rate model with jumps implicitly defines a short-rate-type model with Dirac delta functions so these models have always existed.  Thus it is now possible to choose the most appropriate model setting, short-rate, instantaneous forward rate, or tenor forward rate, without compromising practical expressivity\footnote{Whilst HJM models are theoretically infinite dimensional no practical implementation has more than a very small number of factors.}.  In fact from an engineering point of view short rate models mixing continuous and Dirac processes can be expected to be sufficient as jumps are qualitatively redundant.

We developed pricing using Dirac processes for basic interest rate and single-name credit products, both linear (zero coupon bonds and swaps) and non-linear (options). Dirac processes applied to pricing CDS swaptions  demonstrated that this approach could produce high implied volatilities problematic with other approaches \citep{Jamshidian2004a,Brigo2006a,Kokholm2010a,Brigo2010el,Roti2013a,Weckend2014a,Stamm2015a}.  In general a process $P_t$ with diffusions $M_t$, jumps $A_t$, and Dirac processes $D_t$ can be written as:
\beq
P_t = M_t + A_t + D_t.
\eeq
However, in short rate modelling the fact that jumps add no qualitatively different behaviour (after integration) suggests that jumps may be redundant whilst Dirac processes expand the expressivity of the short-rate, and hazard-rate , modelling paradigms.

Requiring Dirac processes in a short rate setup means that there are jumps in forward rates and also in swap rates.  Thus Dirac process pricing potentially provides an additional tool for the detection of incomplete markets.  However, this approach only requires a single implied volatility unlike methods which require a whole volatility smile to fit (i.e. using tail shape information).  Given that the historical (and implied) CDS volatilities are high enough to make Dirac process modelling attractive this suggests that there are jumps in CDS swap rates and hence that the market is incomplete.  Modelling using integrated survival probabilities \citep{Peng2008a} also required jumps, supporting this implication.   Now if a market is incomplete then no option hedger can be profit-and-loss (PnL) flat in all states of the world (i.e. all outcomes).  Any rational participant will require compensation for potential losses, and PnL volatility.  This moves pricing into a combined real-world and risk-neutral pricing setup \citep{Kenyon2015a} or, equivalently, specialized measures \citep{Cont2003a}.  The required compensation may be higher than the market is willing to pay, thus restricting liquidity as observed \cite{Carver2013c}.  Of course liquidity can be low for many other reasons as well.

Dirac processes are a subset of Generalized Processes that have not previously been directly used in derivative pricing \citep{Seydel2012a} because the dollar value of a Generalized value is meaningless.  For example what is the dollar value of a Dirac delta function at the origin?  This is not even a value in \R.  Generalized processes have a natural connection with stochastic partial differential equations, but this is a separate research area.  Stochastic backward differential equations are different again.  

Any derivative modelling framework can use Dirac processes, either directly in short rate models, or via a time-integral approach otherwise. The method to include them depends on whether the underlying being modelled is a tradable or not.  Short-rate models are based on a non-tradable that must be integrated (at least) before it becomes a tradable so Dirac processes fit naturally.  For models of tradables the Dirac process must be integrated, possibly over very short periods, before incorporation.  For example Dirac processes fit naturally into power pricing for modelling short interval contracts which can go down to five-minute slots.  Non-storable commodities, or commodities with no local storage naturally exhibit spike behaviour.


\begin{doublespacing}   
\bibliographystyle{jf}
\bibliography{kenyon_general}
\end{doublespacing}

\clearpage
\end{document}